\documentclass[twoside,leqno,twocolumn]{article}

\usepackage{ltexpprt}

  \usepackage{amsmath} 
  \usepackage{amsfonts} 
  \usepackage{tikz} 
  \usepackage{numprint}
  \usepackage{xspace}
  \usepackage{hyperref} 

\begin{document}

\makeatletter
\let\@fnsymbol\@arabic
\makeatother

\title{A new Integer Linear Program for the Steiner Tree Problem with Revenues, Budget and Hop Constraints\thanks{Partially supported by DFG, RTG 1855.}}
\author{Adalat Jabrayilov\thanks{TU Dortmund University, Germany.} \\
 \and 
Petra Mutzel\thanks{TU Dortmund University, Germany.}}
\date{}

\maketitle







\def\stprbh{\textrm{STPRBH}\xspace}
\def\setR{\mathbb{R}\xspace}
\def\setN{\mathbb{N}\xspace}
\def\arc#1#2{(#1,#2)}
\def\pos{\pi}
\def\minimize{\textrm{min}}
\def\subjectto{\textrm{s.t.}}
\def\hs{\hspace{8pt}} 
\def\bc{branch-and-cut\xspace}
\def\BC{\emph{branch-and-cut}\xspace}
\def\bp{\emph{branch-and-price}\xspace}
\def\SL{Sinnl-Ljubi\'c\xspace}
\def\Ljubic{Ljubi\'c\xspace}
\def\popeurc{\pop\textrm{+EURC}}
\def\eurc{\textrm{R}\xspace}
\def\red#1{\textcolor{red}{#1}}
\def\alert#1{\textcolor{red}{#1}}
\def\ie{i.e.,}
\def\eg{e.g.,}
\def\pop{\textrm{POP}}

\begin{abstract} \small\baselineskip=9pt 
The Steiner tree problem with revenues, budgets and hop constraints (STPRBH) is a variant of the classical Steiner tree problem. 
This problem asks for a subtree in a given graph
with maximum revenues corresponding to its nodes, where its total edge costs respect the given budget, and the number of edges between each node and its root does not exceed the hop limit. We introduce a new binary linear program with polynomial size based on partial ordering, which (up to our knowledge) for the first time solves all STPRBH instances from the DIMACS benchmark set to optimality. The set contains graphs with up to 500 nodes and \numprint{12500} edges.
\end{abstract}

\section{Introduction}
Many network design applications ask for a minimum cost subtree
connecting some required nodes of a graph.
These applications can be modelled as
the \emph{Steiner tree problem (STP)}:
Given a weighted undirected graph $G$
with node set $V(G)$, edge set $E(G)$, 
edge costs $c \colon E(G) \rightarrow \setR^+$ 
and a subset of the required nodes called \emph{terminals}, 
this problem asks for a subtree $T$ of the graph, 
which contains all terminals and has minimum costs,
\ie~$\sum_{e \in E(T)} c_e $ is minimal.
STP belongs to the classical optimization problems and
is NP-hard \cite{doi:10.1137/0132072}.
The \emph{Steiner tree problem with revenues, budgets and hop constraints} 
(\stprbh)
is a variant of 
the STP
and considers the safety of the connection in addition to the costs.
It originates from telecommunication 
and requires that the constructed tree contains a given service provider 
(root),
and that the path from the provider to each node of the tree has at most
$H$ hops (edges).
The hop limit is needed to control the failure 
of the service, since
the failure probability of the path 
with at most $H$ edges
does not exceed $1-(1-p)^H$,
where $p$ is the failure probability of any edge.
The \stprbh\ is formally defined as follows:
In addition to the edge costs we are given 
a root node $r$, 
the node revenues $\rho \colon V(G) \rightarrow \setR^+$,
the budget $B \in \setR^+$,
and the hop limit $H \in \setN^+$.
The goal is to construct a subtree $T$ 
of the graph,
which contains $r$, maximizes the collected revenues 
$\sum_{v \in V(T)} \rho_v$
and respects the hop and budget constraints, \ie\ 
the number of edges between the root $r$ and each node $v \in V(T)$ does not 
exceed the hop limit $H$ and 
the total edge
costs of the tree respect the budget $B$, \ie\ $\sum_{e \in E(T)} c_e \le B$.
We will call a feasible solution of the \stprbh\ a \emph{Steiner tree}.
Notice that in the literature this term is mostly used for
a feasible solution of the STP.
%

%
The \stprbh\ problem has been introduced 
by Costa et al.~\cite{article}.
They also presented three \bc\ approaches 
based on the Dantzig-Fulkerson-Johnson subtour elimination constraints,
the Miller-Tucker-Zemlin (MTZ) constraints, 
and the Garcia-Gouveia Hop formulation.
Their evaluation shows that the last formulation solves 
the majority of the 
DIMACS benchmark \cite{DIMACS11BENCHMARK} instances 
with up to 500 nodes and 625 edges
within the time limit of two hours. 
However, 
according to the authors \cite{COSTA200868}, 
these algorithms cannot solve even the root relaxation
for most of the large instances with 500 nodes and \numprint{12500} edges.
Therefore, the authors presented three heuristics 
based on the \emph{greedy} method, the \emph{destroy-and-repair} 
method and \emph{tabu search}.
In \cite{sinnl11} Sinnl has introduced
two \bp\ algorithms based on \emph{directed} and \emph{undirected path formulations}
and presented the computational results for all instances with up to 500
nodes and 625 edges. 
His approach solves 
the majority of these instances 
within the time limit of \numprint{10000} seconds.
Layeb et al.~\cite{Layeb2013SolvingTS} have proposed
two new models, one based on the MTZ formulation and one
based on the \emph{reformulation-linearisation-technique}.
In the computational experiments 
they have considered instances with up to 500 nodes and 625 edges, 
and hop limits $3,6,9,5,15$ (instances with $H \in \{12,25\}$ have been ignored).
The experiments show that
their algorithms can solve 
all the considered instances
within the time limit of two hours.
Fu and Hao have introduced two new heuristics,
the \emph{breakout local search} algorithm \cite{DBLP:journals/eor/FuH14}
and
the \emph{dynamic programming driven memetic search} 
algorithm \cite{2015:DPD:2748889.2748890}.
They also presented computational results for the large DIMACS graphs
with up to \numprint{12500} edges, which show an improvement 
of the feasible solutions compared to 
those of 
the heuristics presented in \cite{COSTA200868}.
Recently, Sinnl and \Ljubic \cite{DBLP:journals/mpc/SinnlL16} 
have suggested a \bc\ algorithm based on \emph{layered graphs}.
The algorithm has won the category STPRBH in the DIMACS 
challenge \cite{DIMACS11}.
Up to our knowledge this is the best 
state-of-the-art algorithm for the \stprbh.
While the previous exact algorithms 
can consistently solve only the small instances 
up to 500 nodes and 625 edges and hop limit 15 within the time limit of two hours,
this algorithm solves the majority of the graphs 
up to 500 nodes, \numprint{12500} edges and hop limit 25 
within a time limit of 20 minutes.
For example, the approaches 
\cite{article}, 
\cite{Layeb2013SolvingTS} and
\cite{sinnl11} 
solve the instances with $(|V|,|E|,H)=(500,625,15)$ on average in 
37.78, 104.17 and 127.63
seconds respectively,
while \SL's algorithm needs just 6.46 seconds for this.

\noindent\textit{Our contribution.}
Many graph problems 
can be seen as \emph{partial ordering problems} (\pop), 
\ie~compute a partial ordering of the nodes 
for a given graph 
that minimizes some objective function corresponding to this 
ordering.
Integer linear programming (ILP) formulations based on 
partial orderings
have shown to be practically successful 
for graph drawing 
\cite{Jabrayilov2016}
and vertex coloring \cite{VCPPOP}.
In this paper we present a new ILP based on 
partial ordering for the \stprbh.
In contrast to the algorithm of \SL
it has polynomial size,
and hence has the advantage that it
can be fed directly into a standard ILP solver,
while the former is a sophisticated \bc algorithm, which uses 
an exponential number of 
\emph{subtour elimination constraints}.
We also present an experimental comparison of both approaches
using the DIMACS instances.
While \SL left four DIMACS instances unsolved
within a time limit of 3 hours,
our approach 
solves
all 414 instances within a time limit of 2 hours.
The new approach 
solves 410 of 414 instances within a time limit of 1027 seconds.
Our experiments showed that for all the tested instances
the strength of the LP relaxation of our basic model
dominates the \SL\ basic model.
We also suggest
a new reduction technique for the \stprbh,
which decreases the running times on average 
up to 1.55 times for the largest benchmark graphs
and 1.96 times for instances with the largest hop limit.

\noindent\textit{Outline.}
The paper is organized as follows.
We start with some notations (\autoref{sec:notations}). 
In \autoref{sec:popilp} we present our ILP and 
in \autoref{sec:preprocessing} a new reduction technique for \stprbh. 
The computational results are presented 
in \autoref{sec:evaluation}. 
We conclude with \autoref{sec:conclusion}.

\section{Notations}
\label{sec:notations}
For a graph $G=(V,E)$ we denote its node set by $V(G)$, and 
its edge set by $E(G)$.
Each edge of an undirected graph 
is a 2-element subset $e=\{u,v\}$ of $V(G)$.
For clarity we may write it as $e=uv$.
For an edge $e$ we denote with $G \setminus e$ the resulting graph after 
removing $e$ from $G$.
The end nodes $u,v$ of an edge $uv$ are called neighbours.
With $N(v)$ we denote the set of neighbours of node $v$ in $G$.
Each edge of a directed graph is an ordered pair $e=\arc uv$ of nodes 
and is called a directed edge or arc. 
An arc $\arc uv$ is an outgoing arc of $u$ and an incoming arc of $v$.

For a subgraph $G'$ of $G$
we denote with
$c(G')$ its total edge costs, \ie~$c(G')=\sum_{e \in E(G')} c_{e}$.
The (undirected) path $P$ in graph $G$ is a sequence 
$v_0,e_1,v_1,\cdots,v_{k-1},e_{k},v_k$ of distinct nodes
and edges, 
so that 
$e_i=\{v_{i-1},v_i\}$
for $1\le i \le k$.
We may call this path a $(v_0,v_k)$-path.
Similarly, if $P$ is a directed path, each arc $e_i$ satisfies $e_i=\arc {v_{i-1}}{v_{i}}$.
We denote the length of the unweighted shortest $(u,v)$-path in the input graph $G$ 
with $len(u,v)$.

A tree $T$ is a graph, 
which has exactly one $(u,v)$-path
for any pair of nodes $u,v \in V(T)$.
A rooted tree $T$
has a special node $r \in V(T)$.
The depth $d_v$ of node $v$ in the rooted tree $T$ is the
number of edges in
the $(r,v)$-path in $T$.
The depth of $T$ is the largest depth in it, 
\ie~$\max \{d_v \colon v \in V(T)\}$.

%
\section{Partial-ordering based binary linear program}
\label{sec:popilp}

\subsection{Basic model.} A rooted tree $T$ with depth $h$ 
induces some partial orderings of its nodes.
For example consider an ordering $\pos$
with positions $0,1,\cdots,h$, so that
each node $v$
at depth $d_v$ 
in the tree is
at position $\pos_v=d_v$ 
in the ordering.
The nodes at the same depth are not ordered in $\pos$, while every two 
nodes at different depths are ordered.
In this sense 
we can interpret the \stprbh\ as 
a partial ordering problem 
and model it with the following two sets of 
binary 
variables. 
%
The first set describes the position of each node in the 
ordering, \ie\ 
for each node $v \in V$ and position $i \in \{0,1,\cdots,H\}$ we define two 
variables: 
\begin{align*}
  g_{i,v} =
  \left\{
  \begin{array}[2]{ll}
    1 & \mbox{position of $v$ is greater than $i$, \ie\ } i < \pos_v \\
    0 & \mbox{otherwise} 
  \end{array}
  \right.
  \\
  l_{v,i} =
  \left\{
  \begin{array}[2]{ll}
    1 & \mbox{position of $v$ is less than $i$, \ie\ } \pos_v < i \ \ \ \ \ \\
    0 & \mbox{otherwise} 
  \end{array}
  \right.
\end{align*}

These variables have the property
that if node $v$ is at position $i$, 
then its position is neither 
less nor greater than $i$, and thus both variables are 0, \ie\ 
$l_{v,i}=g_{i,v}=0$.
The second set of binary variables describes the edges of the Steiner tree $T$. 
For each edge $uv \in E$ we define the variables:
\begin{gather*}
  x_{u,v} =
  \left\{
  \begin{array}[2]{ll}
    1 & \mbox{ $T$ contains edge $uv$ and } \pos_u < \pos_v \\
    0 & \mbox{ otherwise} 
  \end{array}
  \right.
  \\
  x_{v,u} =
  \left\{
  \begin{array}[2]{ll}
    1 & \mbox{ $T$ contains edge $uv$ and } \pos_v < \pos_u \\
    0 & \mbox{ otherwise} 
  \end{array}
  \right.
\end{gather*}

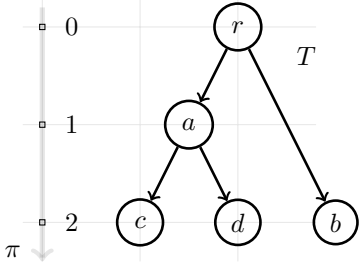
\begin{figure}[bt]
  \centering
  \begin{tikzpicture}[
  x=1.3 cm, y=1.3 cm,
  V3/.style= {circle, inner sep=3pt, draw, line width=1pt},
  V4/.style= {circle, inner sep=4pt, draw, line width=1pt},
  P/.style= {rectangle, inner sep=1pt, draw},
  E/.style= {->,draw, line width=1},
  Pointer/.style= {->, draw, opacity=0.5},
]
  \path [draw, opacity=0.1, step=1.3cm] (-0.2,0.6) grid (2.8,3.3);

  \path (2.7,2.7) node [] {$T$};

  \path (2,3)	node [V4] (r) {$r$};
  \path (1.5,2) node [V4] (a) {$a$};
  \path (3,1)	node [V3] (b) {$b$};
  \path (1,1)	node [V4] (c) {$c$};
  \path (2,1)	node [V3] (d) {$d$};

  \path (r) edge [E] (a);
  \path (r) edge [E] (b);

  \path (a) edge [E] (c);
  \path (a) edge [E] (d);

  \path (0.0,3.2) edge [Pointer, line width=2pt, opacity=0.1] +(0,-2.6) 
  +(-0.3,-2.5) node {$\pos$};

  \path (0, 3) node [P] (p1) {} +(0.3,0) node {$0$};
  \path (0, 2) node [P] (p2) {} +(0.3,0) node {$1$}; 

  \path (0, 1) node [P] (p3) {} +(0.3,0) node {$2$}; 

\end{tikzpicture}
  \caption{$T$ and $\pos$ with $\pos_v \ge d_v$ for each $v\in V(T)$, 
  \eg\ $\pos_r=d_r=0$, $\pos_b=2 \ge d_b=1$, etc.} 
  \label{fig:depth_pos}
\end{figure}

The intuition behind these variables is the following:
According to our basic notions 
an edge $uv$ with some ordering $\pos_u < \pos_v$ can be seen 
as directed arc $\arc uv$.
With this considerations 
$x_{u,v}$ shall be $1$ if and only if $T$ contains an arc $\arc uv$. 

In our construction, we describe $T$ using edge variables $x$ only, 
\ie\ $T$ consists of the edges $\{u,v\}$ with $x_{u,v}=1$ and the nodes contained
in these edges.
Moreover,
for our construction it is enough to have a partial ordering $\pos$, so that
for each node $v$ at depth $d_v$ in $T$, $\pos_v \ge d_v$ 
(Figure \ref{fig:depth_pos}), \ie\ we do not require $\pos_v = d_v$.
With the binary variables $x,l,g$
we formulate our basic model (\pop) for the \stprbh:
\begin{gather}
   \max \hs \rho_{r} + \sum_{uv \in E} (x_{u,v}\cdot\rho_v  + x_{v,u}\cdot\rho_u )
     \tag{\pop}
     \label{model:pop}
     \\
     \mbox{ subject to } \notag
     \\
    l_{r,0} = g_{0,r} = 0
    \label{constr:root}
    \\
    l_{v,1} = g_{H,v} = 0
    ,\ \forall v \in V \setminus \{r\}
    \label{constr:interval}
    \\
    g_{i,v} - g_{i+1,v}  \ge 0
    ,\ 
    \forall v \in V,\ i = 0,\cdots,H-1
    \label{constr:unambiguous1}
    \\
    g_{i,v} + l_{v,i+1} = 1 
    ,\ \forall v \in V,\ i = 0,\cdots,H-1
    \label{constr:unambiguous2}
    \\
    l_{u,i} + g_{i,v} - x_{u,v} \ge 0
    ,\ \forall uv \in E,\ i = 0,\cdots,H
    \label{constr:edge:forward} 
    \\
    l_{v,i} + g_{i,u} - x_{v,u} \ge 0
    ,\ \forall uv \in E,\ i = 0,\cdots,H
    \label{constr:edge:reverse}  
    \\
    \sum_{u \in N(v)} x_{u,v}  \le 1
    ,\ \forall v \in V \setminus \{r\}
    \label{constr:indeg} 
    \\
    \sum_{u \in N(v) \setminus \{w\}} x_{u,v}  \ge x_{v,w}
    ,\  \forall v \in V \setminus \{r\}, w \in N(v)
    \label{constr:outdeg} 
    \\
    \sum_{uv \in E} c_{uv}(x_{u,v}+x_{v,u})  \le B
    \label{constr:budget}
\end{gather}
%
\begin{lemma}
  Let $l,g \in \{0,1\}^{|V|(H+1)}$ and $x \in \{0,1\}^{2|E|}$ 
  be vectors satisfying
  (\ref{constr:root})--(\ref{constr:budget}). Then $x$ describes a tree $T$, 
  which respects the budget. Moreover, if $x \ne 0$, then $T$ contains $r$.
  \label{lemma:budgettree}
\end{lemma}

\begin{proof}
  If $x = 0$, then $T$ is the empty tree and hence respects the budget. 
  Assume $x \ne 0$.

The equations (\ref{constr:root}) ensure that the root $r$ is at position 0.
The remaining nodes $V\setminus \{r\}$ must be placed between the positions 1 and $H$. 
Constraints (\ref{constr:interval}) take care of this.

By transitivity, if a position of a node is greater than $i+1$ then it is also 
  greater than $i$ (constraints (\ref{constr:unambiguous1})).
 Constraints  (\ref{constr:unambiguous2}) express that each node $v$ is either
  at a position greater than $i$ (\ie\ $g_{i,v}=1$) 
  or less than $i+1$ (\ie\ $l_{v,i}=1$) and not both.
  These constraints jointly with constraints (\ref{constr:unambiguous1}) 
    ensure that each node $v$ will be placed at exactly one position, \ie\ 
    there is no position pair $i \ne j$ with 
    $l_{v,i}=g_{i,v}=0$ and 
    $l_{v,j}=g_{j,v}=0$. 
    We show this by contradiction. 
    Let $l_{v,i}=g_{i,v}=0$.
    In  the case $j<i$,  
      as $l_{v,i}=0$ 
      we have 
      $g_{i-1,v}=1$ by (\ref{constr:unambiguous2}).
      Therefore we have
      $g_{j,v}=1$ for each $j \le i-1$ by (\ref{constr:unambiguous1})
      which is a contradiction to $g_{j,v}=0$.  
    In the case $j>i$, 
	as $g_{i,v}=0$ we have 
	$g_{j,v}=0$ for each $j \ge i$ by (\ref{constr:unambiguous1}).
	Therefore we have
	$l_{v,j+1}=1$
	by
	(\ref{constr:unambiguous2}) leading to
	$l_{v,j}=1$ for each $j\ge i+1$ which is
	a contradiction to $l_{v,j}=0$.

  Constraints 
  (\ref{constr:edge:forward}) 
  and
  (\ref{constr:edge:reverse}) 
  make sure that for each $uv\in E(T)$ the following expression holds:
  $x_{u,v}=1$ iff $\pos_u < \pos_v$.
  If $x_{u,v}=1$ we have 
  $l_{u,i} + g_{i,v} \ge 1$ for each $i \in \{1,\cdots,H\}$
  by (\ref{constr:edge:forward}).
  It is easy to see that for two nodes $u,v$
  we have $\pos_u < \pos_v$
  if $l_{u,i} + g_{i,v} \ge 1$ for each $i \in \{1,\cdots,H\}$.
  It follows that if $x_{uv}=1$ then $\pos_u < \pos_v$.
  Analog, in case $x_{v,u}=1$, the constraints
  (\ref{constr:edge:reverse}) enforce
  $\pos_v < \pos_u$.
  From 
  (\ref{constr:edge:forward}) 
  and
  (\ref{constr:edge:reverse}) 
  the claim follows.

  Constraints 
  (\ref{constr:indeg})
  and
  (\ref{constr:outdeg})
  jointly with
  (\ref{constr:edge:forward}) 
  and
  (\ref{constr:edge:reverse}) 
  ensure that 
  $T$ contains no cycle and is connected, \ie\ $T$ is a tree.
  Moreover, $T$ contains $r$.

  Constraints 
  (\ref{constr:indeg}) make sure that each node has at most one incoming 
  arc in $T$, which
  jointly with
  (\ref{constr:edge:forward}) and
  (\ref{constr:edge:reverse})
  ensure that $T$ is cycle free.
  We show this by contradiction and assume that it contains a cycle $C$. 
  Let $v$ be the node in $C$
  with the greatest position. 
  It has two incident edges 
  $uv$ and $vw$ in $C$.
  Since $v$ has the greatest position in $C$ we have $\pos_u < \pos_v$ and 
  $\pos_w < \pos_v$ and hence 
  $x_{u,v}=1$ 
  and
  $x_{w,v}=1$ by
  (\ref{constr:edge:forward}) and
  (\ref{constr:edge:reverse}).
  That means
  $v$ has two incoming arcs 
  $\arc uv$ and $\arc wv$ in $T$
  contradicting (\ref{constr:indeg}).
  
  Constraints 
  (\ref{constr:outdeg})
  jointly with
  (\ref{constr:edge:forward}) and
  (\ref{constr:edge:reverse})
  make sure that $x$ describes a component $T$, 
  which is connected and contains~$r$, 
  \ie\ for each node $v \in V(T) \setminus \{r\}$ 
  there is a 
  directed path from~$r$ to this node.
  We show this by contradiction.
  Assume $T$ is not connected.
  Then it contains at least one nonempty component $C$ with $r \notin V(C)$.
  Since $x$ is an edge variable, it does not describe a component 
  with isolated nodes,
  and hence $C$ contains at least one arc $\arc vw$.
  Then $C$
  has also at least one incoming arc $\arc {v'}v$ to node $v$ by
  (\ref{constr:outdeg}).
  Due to (\ref{constr:edge:forward}) the position of $v'$ is less 
  than $\pos_v$, \ie\ $\pos_{v'} \le \pos_v -1$.
  If $\pos_{v'}=0$ then $v'=r$, since $r$ is the only node at position~0.
  This contradicts $r \notin V(C)$.
  Else we can repeat this argument and get 
  an incoming arc $\arc {v''}{v'}$ to $v'$ in $C$.
  Notice that the position of $v''$ is now at most $\pos_v-2$. 
  So after repeating this argument at most $\pos_v$ times we reach 
  the position $0(=\pos_v-\pos_v)$.
  Since~$r$ is the only node at position 0, $C$ contains $r$, a contradiction.

  Finally, the constraints 
  (\ref{constr:budget}) ensure that the costs of $T$ do not exceed
  the budget-limit.
  \hfill\rule{5pt}{5pt}
\end{proof}

\begin{lemma} 
\label{lemma:depthVsPos}
The model (\pop) constructs a tree $T$ and a partial ordering $\pos$ with:
\begin{description}
\item [\textnormal{(a)}]\emph{For each edge $uv$ of $T$ we have $d_u<d_v$ iff $\pos_u < \pos_v$.}
\item [\textnormal{(b)}]\emph{For each node $v$ of $T$ we have $d_v \le \pos_v$.}
\end{description}
\end{lemma}
\begin{proof}
If $T$ is empty, both statements are clearly satisfied. Assume $T$ is not empty.

\textbf{(a)}
  There are 
  two possibilities:
  either $d_u < d_v $ or $d_v < d_u $.
  The possibilities according to $\pos$ are similar:
  either $\pos_u < \pos_v$ or $\pos_v < \pos_u$.
  Therefore, it is sufficient to show one direction, \ie\ 
  $d_u < d_v \Longrightarrow \pos_u < \pos_v$.
  The second direction follows if we reverse the roles of $u$ and $v$.
  We show this by induction over the edges of $T$, whereby we traverse the tree 
  in breadth-first search (BFS) order. 
  Due to Lemma \ref{lemma:budgettree}, $T$ contains the root $r$.
  If $u = r$, then 
  we have $\pos_u=\pos_r=0$ by~(\ref{constr:root}) 
  and $\pos_v \ge 1$ by
  (\ref{constr:interval}), thus $\pos_u < \pos_v$.
  Else $T$ has an edge $wu$ with $d_w < d_u$.
  The edge has been considered already because of BFS, and hence
  $\pos_w < \pos_u$ holds by induction hypothesis, 
  and thus
  $\pos_v < \pos_u$ is excluded by
  (\ref{constr:indeg}).

\textbf{(b)}
  Let $d_v = l$. 
  Then there is a $(r,v)$-path 
  $v_0,\{v_0,v_1\},v_1,\cdots,v_{l-1},\{v_{l-1},v_l\},v_l$ 
  with $l$ edges
  in $T$, where $v_0=r$ and $v_l=v$. 
  Each edge $\{v_{i-1},v_{i}\}$ in this path satisfies
  $d_{v_{i-1}} < d_{v_{i}}$ and thus 
  $\pos_{v_{i-1}} < \pos_{v_{i}}$ 
  by (a), 
  and hence
  $\pos_{v_{i}} - \pos_{v_{i-1}} \ge 1$, since the positions are integers.
  It follows 
  $d_v = l \le 
  (\pos_{v_1}-\pos_{v_0})
  +(\pos_{v_2}-\pos_{v_1})
  +\ldots
  +(\pos_{v_l}-\pos_{v_{l-1}}) 
  =(\pos_{v_l}-\pos_{v_0}) 
  =(\pos_{v}-\pos_{r}) 
  =\pos_{v}$.  

  \hfill\rule{5pt}{5pt}
\end{proof}

\begin{theorem}
  The basic model (\pop) computes an optimal solution to the \stprbh. 
\label{theorem:pop}
\end{theorem}

\begin{proof}
  Due to Lemma \ref{lemma:budgettree}, the basic model (\pop) constructs a tree $T$, 
  which is rooted by $r$
  and respects the budget. 
  $T$ respects also the hop limit $H$, \ie\ the depth $d_v$ of 
  each node $v$ of $T$ is at most $H$.
  This follows from Lemma \ref{lemma:depthVsPos} and 
  constraints (\ref{constr:interval}),
  since we have $d_v \le \pos_v$ by the lemma and 
  $\pos_v \le H$ by the constraints.
  
  So we only need to show that $T$ is optimal.
  Clearly, every feasible solution satisfies the constraints 
  (\ref{constr:root})--(\ref{constr:budget}).  
  If $x=0$, then $E(T) = \emptyset$. In this case
  a tree $T=(\{r\}, \emptyset)$ 
  has value $\rho_r$ and 
  is optimal, 
  since $\rho_r \ge 0$.
  Thus the objective satisfies $\sum_{v \in V(T)} \rho_v = \rho_{r}$. 
  Else $x \ne 0$, \ie\ $T$ has some edges. 
  In this case $r \in V(T)$ by Lemma \ref{lemma:budgettree}.
  %
  Moreover, 
  each node $v \ne r$ in the tree 
  has exactly one incoming arc $\arc uv$,
  and thus 
  we have $\rho_v = \sum_{uv \in E} x_{u,v} \cdot \rho_v$.
  Hence the objective satisfies
  $\sum_{v \in V(T)} \rho_v = 
  \rho_{r} + \sum_{uv \in E} (x_{u,v}\cdot\rho_v  + x_{v,u}\cdot\rho_u ).$
  \hfill\rule{5pt}{5pt}
\end{proof}

\noindent\textit{Model size.}
The basic model (\pop) has $2(H+1)|V| + 2|E|$  
binary variables and $O(|V|H+|E|H)$ constraints. 
Notice that 
the equations 
(\ref{constr:root}),
(\ref{constr:interval})
and 
(\ref{constr:unambiguous2})
can be used to eliminate
all $l$ variables and the variables $g_{0,v}$, $g_{H,v}$ for each $v \in V$.
Moreover, 
(\ref{constr:root})
and 
(\ref{constr:unambiguous1}) fix the variable $g_{i,r}$ 
for each $i\in \{0,\cdots,H\}$.
The number of remaining variables then is $(H-1)(|V|-1) + 2|E|$.

\subsection{Strengthening constraints.}
\label{ssec:otherconstraints}
As one can see the depth $d_v$ of a node $v$ in $T$ is at least $len(r,v)$.
This is used in 
\cite{DBLP:journals/mpc/SinnlL16}
to fix some variables.
We can apply the idea as follows:
From Lemma \ref{lemma:depthVsPos} 
follows $len(r,v) \le d_v \le \pos_v$, 
\ie\ 
$\pos_v > len(r,v)-1$.
Hence we set  
for each $v \in V \setminus \{r\}$ with $len(r,v)\le H$:
\begin{align}
g_{len(r,v)-1,v} =1
\label{constr:lenpos}
\end{align}
The equations imply also $g_{i,v}=1$ for each $i < len(r,v)-1$
by (\ref{constr:unambiguous1}).

Let $v$ be a node of $T$ with $\pos_v=H$.
Then there is no edge $vw$ in $T$ with $d_w > d_v$, otherwise it would be
$\pos_w > \pos_v$ by Lemma \ref{lemma:depthVsPos} and thus $\pos_w \ge H+1$.
Hence $v$ is a leaf node in $T$.
If the revenue of this node is $\rho_v=0$ we can remove it from $T$ 
without changing the objective value of $T$. 
Therefore we can require from each node $v$ with 
revenue $\rho_v=0$ and $len(r,v) \le H-1$
that 
$\pos_v \le H-1$, \ie\ we set
for each $v \in V$ with $\rho_v=0$ and $len(r,v)\le H-1$:
\begin{gather}
g_{H-1,v}=0.
\label{constr:nonterminal-innode}
\end{gather}

Let $H \ne 0$ and $M:=\{ rv \in E \colon c_{r,v} \le B\} \ne \emptyset$.
Then the following constraint is valid:
\begin{gather}
\sum_{e \in M} x_e \ge 1.
\label{constr:deltaroot}
\end{gather}

Consider some leaf $v$ with $\rho_v>0$ at depth $d_v < H$ of $T$. 
According to Lemma \ref{lemma:depthVsPos}, its position $\pos_v$ holds 
$d_v \le \pos_v \le H$. To break this type of symmetries we require 
$\pos_v=H$, \ie\ 
for each $v \in V \setminus \{r\}$  with $\rho_v>0$ we have:
\begin{gather}
\sum_{w \in N(v)} x_{v,w} \ge 1-g_{H-1,v}.
\label{constr:terminalleaf}
\end{gather}
If $v$ is a leaf node 
then it has no outgoing arc $\arc vw$ and thus the left hand side
is 0. This forces $g_{H-1,v}=1$ as desired on the right hand side.
Notice that 
(\ref{constr:terminalleaf})
set $\pos_v=H$ also for each
node $v \in V \setminus V(T)$ with $\rho_v>0$ and removes some more symmetries.

The number of the strengthening inequalities is at most $|V|$.
The constraints 
(\ref{constr:lenpos})
and 
(\ref{constr:nonterminal-innode})
are equations, which fix some variables.


\subsection{Comparison with the \SL model.} 
The \SL model \cite{DBLP:journals/mpc/SinnlL16} 
constructs a Steiner tree~$T'$, which is rooted at $r$ and is 
a subtree of the \emph{layered graph} 
$G_L$, 
whereby
$V(G_L)=V_0 \cup\cdots\cup V_H$ with $V_0=\{r\}$ and 
$V_1=\cdots=V_H=V \setminus \{r\}$, \ie\ 
in $G_L$ the root $r$ is on layer 0, and nodes $V_i$ 
on layers $i \in \{1,\cdots,H\}$.
There is an edge only between the nodes in consecutive layers, 
\ie\ $E(G_L)=E_1 \cup \cdots \cup E_H$, where for each edge $uv\in E$ and 
for each $i \in \{1,\cdots,H\}$,\ 
$E_i$ contains the edge between $u \in V_{i-1}$ and $v \in V_i$.
Moreover, 
a node $v$ at depth $i$ in $T'$ is selected from $V_i$, where $1\le i \le H$,
and thus the tree respects the hop limit $H$. 
To describe the layers of the selected nodes,
the \emph{assignment variables} $y_{v,i}$ are defined
for each node $v \in V \setminus \{r\}$ and layer $i\in \{1,\cdots,H\}$, 
which are 1 if $v$ is selected from  $V_i$, and 0 otherwise.
To indicate whether a node or an edge is a part of $T'$ 
additional $|V|+2|E|$ binary variables are used.
%
%
The algorithm contains a basic ILP with a polynomial number of constraints,
which has been enlarged by additional exponential number of 
\emph{subtour eliminations constraints}, and solve the resulting ILP 
with a sophisticated \bc technique.

%

  There is the following connection between both approaches.
  We can interpret the layers as positions of a partial ordering. 
  To describe the positions, 
  the \SL model uses the assignment variables $y$,
  while our approach uses
  the \pop\ variables $l,g$, whereby
  if a node $v$ is on position $i$, then $y_{v,i}=1$, but 
  $l_{v,i}=g_{i,v}=0$.
  Moreover, the orderings $\pos'$ and $\pos$, which are constructed by 
  \SL and \pop, respectively, have the following difference.
  For each node $v$ at depth $d'_v$ in $T'$, $\pos'$ satisfies $\pos'_v=d'_v$,
  while for each node $u$ at depth $d_u$ in $T$, we have $\pos_u \ge d_u$
  by Lemma \ref{lemma:depthVsPos}, 
  which makes 
  the use of the symmetry breaking constraints (\ref{constr:terminalleaf}) possible.

As we mentioned above, our model has $(H-1)(|V|-1) + 2|E|$ binary variables
and $O(|V|H+|E|H)$ constraints,
while the \SL\ model has $H(|V|-1)+|V|+2|E|$ binary variables, and an exponential number of
constraints.

\section{Reducing the problem size}
\label{sec:preprocessing}
To reduce the problem size we introduce
a new preprocessing technique with running time~$O(|V|^2)$,
which extends
the \emph{undirected root cost (URC)}
test from \cite{DBLP:journals/mpc/SinnlL16}.

\subsection{Extended undirected root cost test (EURC).}
\label{sec:eurc}
  One of the reduction techniques used 
  by Sinnl and \Ljubic \cite{DBLP:journals/mpc/SinnlL16}
  is the URC reduction, which is described
    as follows:
    Let $u$ and $v$ be two neighbours 
    of the root $r$ so that $uv \in E$. As one can see, if 
    the cost $c_{uv}$ of the edge $uv$
    is higher than 
    the costs of the edges $ru$ and $rv$, then we can remove $uv$ from the graph.
    Since we can connect the root with each of the nodes $u$ and $v$ cheaper
    than $c_{uv}$, we can construct the optimal Steiner tree $T$ without 
    this edge. 
  We extended the URC test as follows:

\begin{lemma}[EURC test]
\label{lemma:eurc}
    Let $uv$ be some edge in $G$ with $u \ne r$ and $v \ne r$, and let
    $P_{r,u}$ and $P_{r,v}$
    be two paths in $G$ 
    from the root to $u$ and $v$, respectively.
    The reduced graph $G \setminus uv$ contains a Steiner tree with 
    the same objective value as the value of an optimal solution $T$ in~$G$, if
    both the following conditions are satisfied:
\begin{align}
  c(P_{r,v}) \le c_{uv} \mbox{ \ \ and \ \ } |E(P_{r,v})| \le len(r,u) + 1
\label{prep:1} \\
  c(P_{r,u}) \le c_{uv} \mbox{ \ \ and \ \ } |E(P_{r,u})| \le len(r,v) + 1
\label{prep:2}
\end{align}
\end{lemma}

\begin{proof}
The intuition behind the test is that these two paths can be used as 
alternatives to the edge $uv$. In this sense we call them \emph{alternate paths}.
If $uv \notin E(T)$
then $T$ is a subtree of $G\setminus uv$,
and we are done. So suppose $uv \in E(T)$.
Assume without loss of generality that in $T$ the depth of $u$ is less than 
the depth of $v$.
In this case we need (\ref{prep:1}), otherwise (\ref{prep:2}).
Removing this edge from $T$, decomposes it into two subtrees $T_1$ and $T_2$.
The cost of $T$ satisfies:
\begin{gather*}
c(T) = c(T_1) + c_{uv} + c(T_2).
\end{gather*}
The union of $T_1, T_2$ and $P_{r,v}$
induces a connected subgraph $C$ of $G$,
which consists of
node set
$V(T_1) \cup V(P_{r,v}) \cup V(T_2)$ 
and edge set
$E(T_1) \cup E(P_{r,v}) \cup E(T_2)$. 
Since
$r \notin \{u,v\}$
the path 
$P_{r,v}$ does not contain the edge $uv$, otherwise $P_{r,v}$ 
would contain at least one other edge $e$ than $uv$ and cost 
$c(P_{r,v}) \ge c_e + c_{uv} > c_{uv}$ since $c_e \in \setR^+$.
The component $C$ contains all nodes of $T$, since 
$V(T)=V(T_1) \cup V(T_2) \subseteq V(C)$. 
We show that $C$ contains a subtree $T'$, 
which contains 
all nodes of $T$,
and respects the budget and hop constraints.
According to the budget constraint, we show that any subtree $T'$ of $C$
does not cost 
more than $T$:
\begin{align*}
c(T') &\le 
c(C)=c(T_1) + c(P_{r,v}) + c(T_2) \\
     & \le c(T_1) + c_{uv} + c(T_2) = c(T).
\end{align*}
Hereby, the second inequality follows from (\ref{prep:1}), \ie\ $c(P_{r,v}) \le c_{uv}$.
According to the hop constraint, we show that
for each node $w$ at depth $k$ in $T$ 
the component $C$ has a $(r,w)$-path with at most $k$ edges.
Let $W_{r,w}$ be the $(r,w)$-path in $T$, \ie\  
$|E(W_{r,w})|=k$.
If $w$ is in $T_1$, then 
this path 
is also 
in $C$, since $E(T_1) \subseteq E(C)$. 
So suppose $w$ is in $T_2$. 
The path $W_{r,w}$ 
consists of a subpath $W_{r,u}$ from $r$ to $u$, the edge $uv$ and 
a subpath $W_{v,w}$ from $v$ to $w$.
Since
$|E(W_{r,u})| \ge len(r,u)$
we have
$k=|E(W_{r,w})| \ge len(r,u)+1 + |E(W_{v,w})|$.
The $(r,w)$-path in $C$ consisting of 
the subpath 
$P_{r,v}$ and $W_{v,w}$ has at most
$len(r,u)+1 + |E(W_{v,w})| \le k$ edges, as $|E(P_{r,v})| \le len(r,u) + 1$ 
by (\ref{prep:1}). 
  \hfill\rule{5pt}{5pt}
\end{proof}
%

\subsubsection{Applying EURC.}
Let $e$ be an edge, which has alternate paths in $G$. 
According to lemma \ref{lemma:eurc}
after removing $e$ from $G$, the reduced graph $G'$ 
contains an optimal solution, too.
Thus we can apply the EURC reduction on $G'$ as well. 
In this way we can iterate over all edges of $G$ once
and remove each edge, which has alternate paths.
More precisely, 
we compute for each $v \in V$ one weighted shortest $(r,v)$-path in $G$,
where every edge $e \in G$ has weight $c_e$.
We then iterate over all edges of $G$ once
and remove the edge $uv$ if
the $(r,u)$-path and $(r,v)$-path are alternate paths to $uv$.
This takes 
$O(|E| + |V|\log|V|)$ 
time if we use the Dijkstra algorithm to 
compute the single source shortest paths.

Suppose we apply the EURC test on a reduced graph $G'$ after some edge removal. 
%
Since the test depends on the shortest paths in $G'$,
we need to show for correctness that
for each $v \in V$ the weighted shortest $(r,v)$-paths in $G$ exist in $G'$, too
(Lemma \ref{lemma:eurc:weightedpath}),
and that the length of the unweighted  shortest $(r,v)$-path in $G'$ is the same as
$len(r,v)$ (Lemma \ref{lemma:eurc:unweightedpath}).

\begin{lemma}
  \label{lemma:eurc:weightedpath}
  Let $v \in V$ and let $G'$ be a reduced graph after removing some edges 
  by EURC.
  The weighted shortest $(r,v)$-path $P_{r,v}$ 
  in $G$ exists in $G'$, too.
\end{lemma}
\begin{proof}
Assume $P_{r,v}$ does not exist in $G'$
and there is an edge $ww' \in E(P_{r,v}) \setminus E(G')$. 
W.l.o.g. $P_{r,v}$ consists of a $(r,w)$-path $P_1$, an edge $ww'$
and a $(w',v)$-path $P_2$.
Since EURC does not remove an edge $ww'$ if $r \in \{w,w'\}$,
we have $|E(P_1)|>0$ and thus $c(P_{1})>0$.
Hence we have 
$c(P_{r,v})=c(P_1)+c_{w,w'}+c(P_2) > c_{w,w'}+c(P_2)$.
As $ww'$ is removed by EURC, there is a 
$(r,w')$-path
$P_{r,w'}$ in $G$ satisfying 
$c(P_{r,w'}) \le c_{ww'}$.
The paths $P_{r,w'}$ and $P_2$ build a $(r,v)$-path in $G$, which costs
$c(P_{r,w'})+c(P_2) \le c_{w,w'}+c(P_2) < c(P_{r,v})$ contradicting
that $P_{r,v}$ is a weighted shortest $(r,v)$-path in $G$.
  \hfill\rule{5pt}{5pt}
\end{proof}

\begin{lemma}
  \label{lemma:eurc:unweightedpath}
  Let $v \in V$. Removing the edges by EURC does not change
  the length of an unweighted shortest $(r,v)$-path.
\end{lemma}
\begin{proof}
Let $uw$ be the edge, so that after its removal 
the length of an unweighted shortest $(r,v)$-path 
$len'(r,v)$ in the reduced graph $G'$ has changed, \ie\ 
$len'(r,v) \ne len(r,v)$. 
Let $P$ be a $(r,v)$-path in $G$ with $|E(P)|=len(r,v)$, which has been destroyed
by removing the edge 
$uw \in E(P)$.
W.l.o.g. $P$ consists of a $(r,u)$-path, the edge $uw$
and a $(w,v)$-path $W$.
We have $len(r,v)=|E(P)|=len(r,u) +1 + |E(W)|$.
As $uw$ has been removed by EURC there is a path $P_{r,w}$ in $G$ with 
$|E(P_{r,w})| \le len(r,u) + 1$ by (\ref{prep:1}).
Due to lemma \ref{lemma:eurc:weightedpath} this path exists in $G'$. 
The paths $P_{r,w}$ and $W$ build a $(r,v)$-path in $G'$
with
length $|E(P_{r,w})| +|E(W)| \le len(r,u) + 1 + |E(W)| = len(r,v)$. 
From $len'(r,v) \le |E(P_{r,w})| +|E(W)|$ it
follows that $len'(r,v) \le len(r,v)$. 
Since $G'$ is a subgraph of $G$ we have $len'(r,v) \ge len(r,v)$,
and thus
$len'(r,v) = len(r,v)$.
  \hfill\rule{5pt}{5pt}
\end{proof}
%

%
\section{Computational Results}
\label{sec:evaluation}
%
We implemented our model with the Gurobi-python API
and the reduction with 
the python library \url{http://networkx.readthedocs.io}.
%
We tested four variants of our model:

\begin{description}
  \item[\pop1] 
  consists of
  the basic model (\pop), 
  (\ref{constr:root})--(\ref{constr:budget})
  and the strengthening constraints 
  (\ref{constr:lenpos})--(\ref{constr:deltaroot}),
  %
  where (\ref{constr:lenpos})
  and 
  (\ref{constr:nonterminal-innode})
  are equations for fixing some variables,
  and 
  (\ref{constr:deltaroot}) is exactly one inequality.

  \item[\pop2] 
  extends \pop1 by the strengthening inequalities (\ref{constr:terminalleaf}).

  \item[\pop1\eurc, \pop2\eurc] are \pop1 and \pop2, respectively,
  after applying the EURC reduction. 
\end{description}
The source code of these models 
is available 
on our benchmark site \cite{POP:BENCHMARK}.
We compared our models 
with the sophisticated \bc algorithm 
suggested by Sinnl and \Ljubic \cite{DBLP:journals/mpc/SinnlL16}.
%
As we mentioned in \autoref{sec:eurc} their algorithm uses another reduction techniques.
Their source code, which is implemented in C++, is publically available, too.
The comparisons were performed with a time limit of 3 hours 
on an 
Intel Core i7-4790, 3.6 GHz, with 32 GB of memory and running Ubuntu
Linux 16.04. 
To solve
our models, we used Gurobi 6.5.1
single-threadedly with parameter MIPGap=$10^{-5}$ 
(default value is $10^{-4}$). 

%
%
For the experiments we used
the DIMACS benchmark set \cite{DIMACS11BENCHMARK}.
The set consists of 414 instances, which are created
based on the B and C classes of the OR-Library \cite{ORLib} for the STP.
Given a STP graph, a \stprbh\ instance has 
the same graph and edge costs. 
The root is selected from the terminals. 
An instance originating from an instance $I$ of B class has 
the file name format $I$-$R$-$H$, e.g. ``B01-5-3.stp'', where
$H$ is the hop limit, and $R$ is a positive integer.
A node $v$ has a randomly selected revenue $\rho_v \in [1,R]$
if it is a terminal node, and $\rho_v = 0$ otherwise.
The budget is $B:=\frac1b \sum_{e \in E} c_e$,
where $b \in \{5,10\}$ by convention.
An instance originating from the C class contains $b$ as well, 
\ie\ it has the file name format $I$-$R$-$b$-$H$. 

%
%
%
%


\begin{figure*}[h]
  \centering
  \begin{tikzpicture}
    \path (-4.3,0) node {\includegraphics[scale=1]{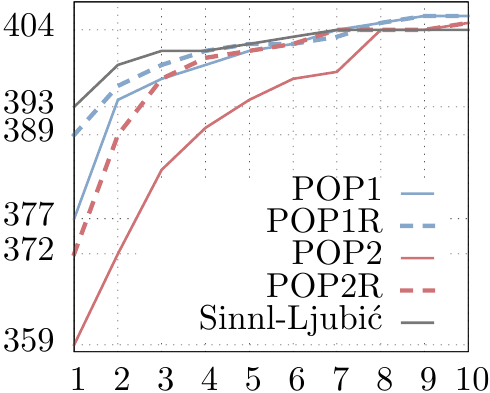}};
    \path (-7.1,0) node [rotate=90] {\#solved instances};
    \path (-4.0,-2.5) node {time limit [min]};

    \path ( 4.4,0) node {\includegraphics[scale=1]{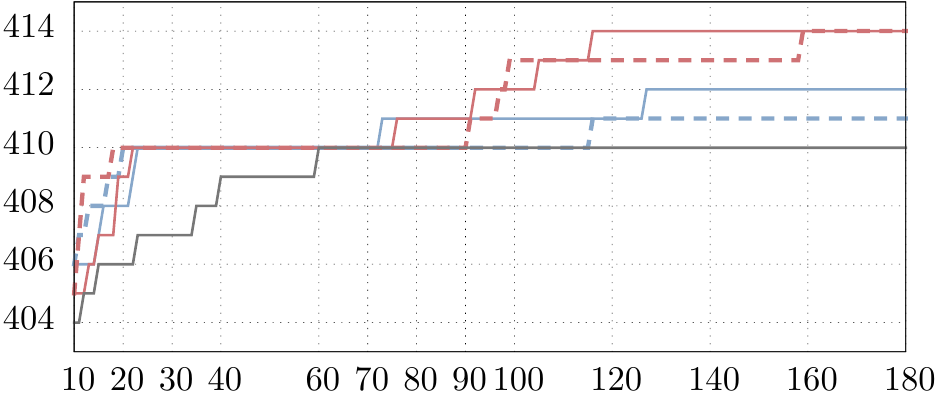}};
    \path (-0.7,0) node [rotate=90] {\#solved instances};
    \path ( 4.5,-2.5) node {time limit [min]};

  \end{tikzpicture}
  \caption{Number of the solved instances by the five models within different time limits} 
  \label{fig:chronology}
\end{figure*}
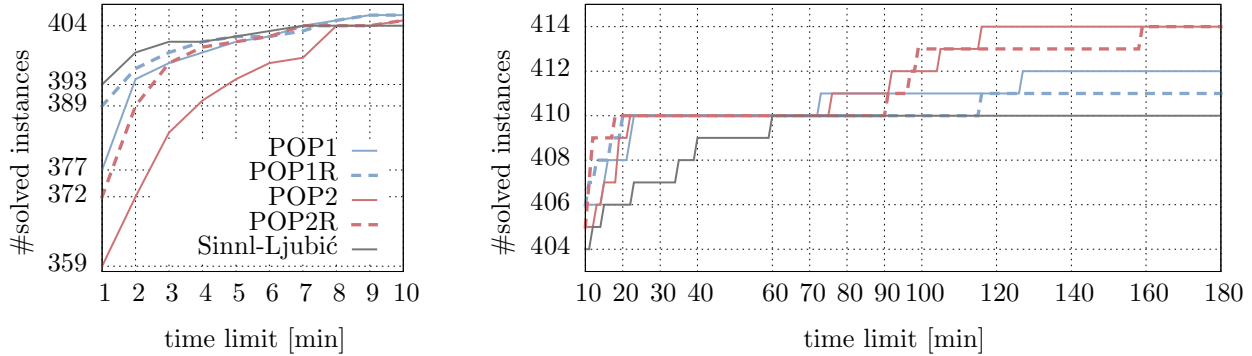

\begin{table*}[bt]
  \setlength{\tabcolsep}{5.4pt} 
  \begin{tabular}{l|rr|r|rrrr|rrr}
    \hline
              &       &       &    &\pop1   &\pop1\eurc &\pop2   &\pop2\eurc &     &     &\SL\\ 
     Instance & $|V|$ &$|E|$  &Opt &Time[s] &Time[s]    &Time[s] &Time[s]    &lb &ub&Time[s] \\ 
    \hline
    \textbf{C10-100-20-15}    &500  &1000   &\textbf{5906}  &488   &427   &465   &378   &5906  &5967  &tl    \\
\textbf{C10-10-20-15}     &500  &1000   &\textbf{573}   &246   &402   &378   &572   &573   &576   &tl    \\
\textbf{C18-100-1000-5}   &500  &12500  &\textbf{3320}  &7574  &6946  &6919  &9486  &3320  &3368  &tl    \\
C18-10-1000-5             &500  &12500  &318            &4366  &tl    &4505  &5881  &318   &318   &3563  \\
\textbf{C20-100-1000-15}  &500  &12500  &\textbf{5222}  &1283  &997   &764   &684   &5222  &5233  &tl    \\
C20-100-1000-5            &500  &12500  &4768           &tl    &tl    &6297  &5417  &4768  &4768  &1328  \\
C20-10-1000-5             &500  &12500  &460            &tl    &tl    &5474  &5802  &460   &460   &667   \\
 
    \hline
     \#unsolved   & & &  &2 &3 &0 &0 & & &4 \\ 
     \hline
  \end{tabular}
  \caption{All DIMACS instances, which remain unsolved by any of the five models}
  \label{tab:resultsDIMACS:hard}
\end{table*}

\subsection{Performance of the ILPs.}
%
Figure \ref{fig:chronology} 
visualizes 
for 
each model 
the number of instances,
which can be solved 
within a time limit of 
$1,2,\ldots,180$ minutes.
%
For the result of each single instance see 
the Appendix.
As we can see, all five models 
can solve the majority of the instances
in a short time.
%
%
Within the time limit of 
a minute 
\pop1,  \pop1\eurc,  \pop2,  \pop2\eurc\  and  \SL\  solve  
377,  389,  359,  372  and  393  instances  respectively.
%
%
%
%
410 of 414 instances take at most 
1027 seconds by \pop2\eurc\
and 
1349,
1192,
1289,
3563 seconds by
\pop1, \pop1\eurc, \pop2, \SL\ respectively.
For the remaining 4 instances, \pop2 and \pop2\eurc need more than an hour
but at most 6919 and 9486 seconds, respectively.
Within the time limit of 3 hours 
\pop1, \pop1\eurc and \SL left 2, 3 and 4 instances unsolved, respectively.
Up to our knowledge, \pop2 and \pop2\eurc\ solve all instances, including 
4 unsolved instances by \SL\ for the 
first time. The results of these 4 instances are shown bold in
Table \ref{tab:resultsDIMACS:hard}. 
The table contains all instances, which remain unsolved by
any of the five models.
%
%
Columns 1--3 show the instance names and sizes.
Column 4 displays the optimal values. 
Columns 5--8 contain the runtimes of our models. 
%
Columns 9--11 display the lower and upper bounds as well as the 
running time of \SL.
The times in the table are given in seconds.
An entry
``tl'' 
indicates that 
the time limit is reached. 
The last row shows the number of the unsolved instances for each model.

Since the size of our ILP depends on the instance properties 
$|V|,|E|$ and the hop limit $H$,
it is interesting to see the runtimes 
depending on these properties. 
For this experiment
we ignored the 7 outliers from Table \ref{tab:resultsDIMACS:hard}
and considered the remaining 407 DIMACS instances.

\begin{figure*} 
  \centering
\begin{tikzpicture}
  \path (-3.0,0) node {\includegraphics[scale=1.0]{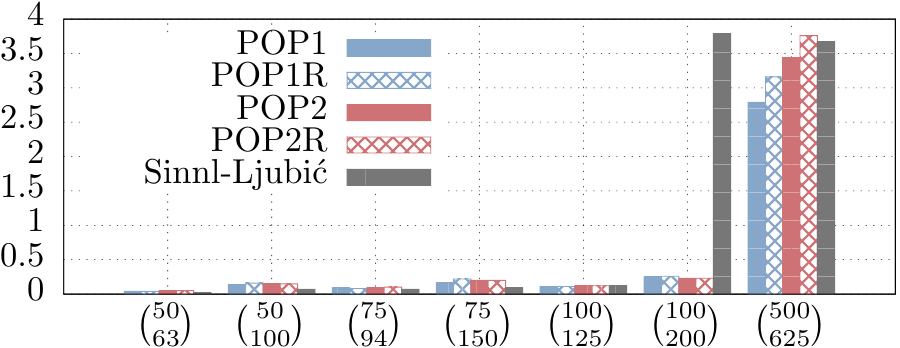}};
  \path (-2.7,-2.2) node {graphs with size ${|V|}\choose{|E|}$};
  \path (-8.0,0.2) node [rotate=90] {average time [s]};

  \path ( 5.5,0) node {\includegraphics[scale=1.0]{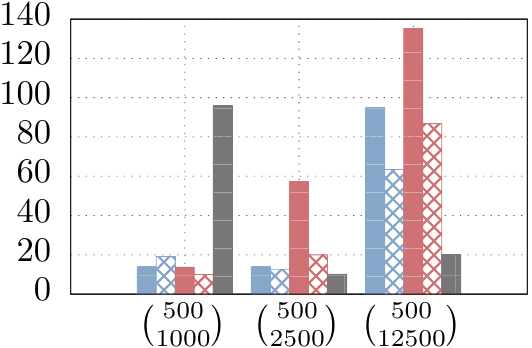}};
  \path ( 5.9,-2.2) node {graphs with size ${|V|}\choose{|E|}$};
  \path ( 2.4,0.2) node [rotate=90] {average time [s]};

\end{tikzpicture}
  \caption{Average running times of the models depending on the graphs sizes}
  \label{fig:histogram:tVE}
\end{figure*}
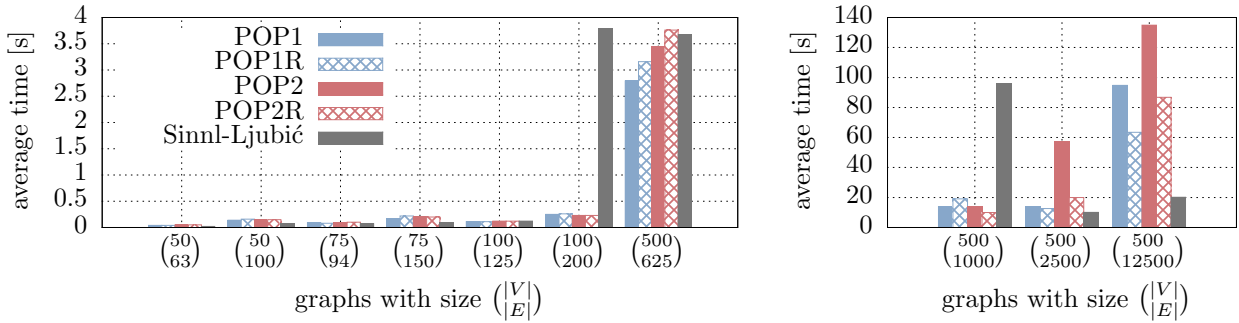

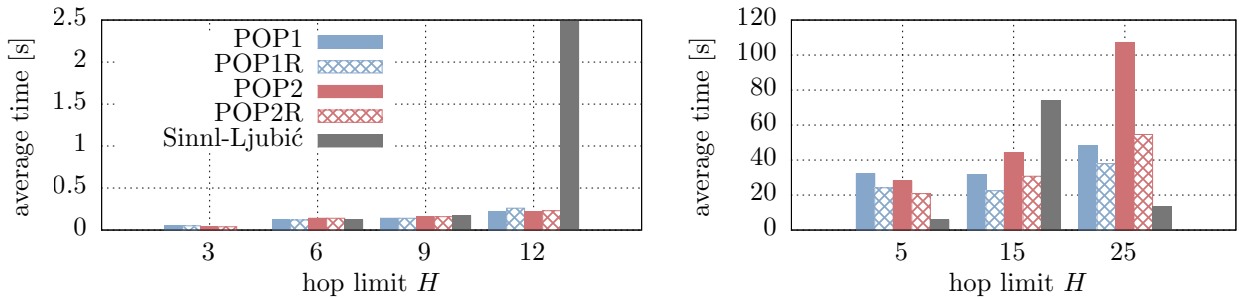
\begin{figure*}[bt]
  \centering
  \begin{tikzpicture}
    \path (-3.0,0.0) node {\includegraphics[scale=1]{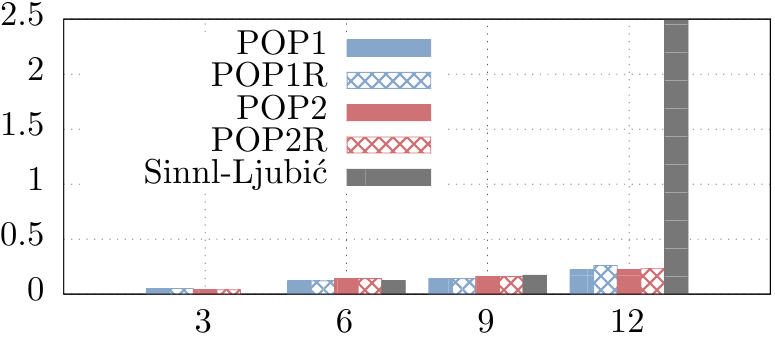}};
    \path (-7.35,0.1) node [rotate=90] {average time [s]};
    \path ( 5.5,0.0) node {\includegraphics[scale=1]{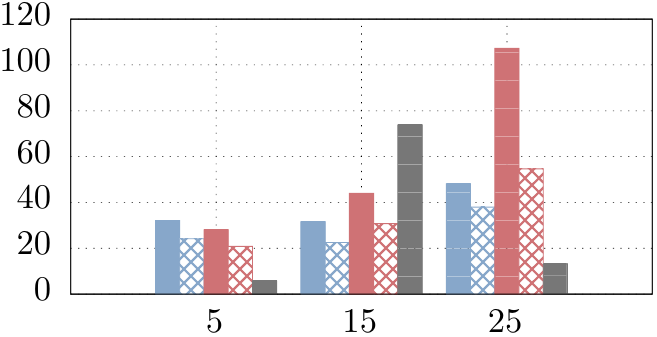}};
    \path ( 1.7,0.1) node [rotate=90] {average time [s]};
    \path (-2.7,-2.0) node {hop limit $H$};
    \path (5.9,-2.0) node {hop limit $H$};
  \end{tikzpicture}
  \caption{Average running times of the models depending on the hope limit $H$}
  \label{fig:histogram:tH}
\end{figure*}

First we consider the dependency of the runtime on the graph sizes.
The DIMACS 
instances have exactly 10 different sizes $(|V|,|E|)$. 
%
We grouped the 407 graphs according to their sizes.
Figure \ref{fig:histogram:tVE} visualizes 
for each group the average runtime of each model in seconds.
For the graphs up to 500 nodes and 625 edges except the group $(100,200)$
the runtimes of all models seem to be similar.
%
For the graphs with sizes 
$(100,200)$ and 
$(500,\numprint{1000})$ 
the new models seem to be faster.
For the graphs with sizes 
$(500,\numprint{2500})$ and
$(500,\numprint{12500})$ 
\SL\ seems to be better.
However, according to Table \ref{tab:resultsDIMACS:hard} 
it does not solve all of the instances with sizes $(500,\numprint{12500})$.
%
%
%
With respect to the worst average runtime over all 10 groups
\pop1\eurc (63.46 seconds for group (500,\numprint{12500})) 
is the fastest, 
and \pop2\eurc (86.82 seconds for group (500,\numprint{12500})) the second fastest model.
However, according to Table \ref{tab:resultsDIMACS:hard} 
\pop1\eurc does not solve all instances.
\pop1\eurc with 18.3 seconds has also the shortest average runtime over all 
407 graphs, while \pop1, \pop2, \pop2\eurc, \SL need on average 
24.2, 39.1, 23.1, 20.1 seconds, respectively.
We can see also the advantage of the EURC reduction for large graphs:
%
%
%
%
%
For the groups (500,\numprint{2500}) and (500,\numprint{12500}) 
it speeds up \pop1, 
respectively, 1.1(=13.82s/12.62s) and 1.49(=94.76s/63.46s) times, 
and \pop2
respectively 2.86(=57.17s/19.96s) and 1.55(=134.88s/86.82s) times on average.
%

%

We now consider the dependency of the runtime on the hop limits $H$.
The DIMACS instances have exactly 7 different hop limits. 
We grouped the 407 instances according to their hop limits.
Figure \ref{fig:histogram:tH} visualizes 
the average runtime of each model in seconds for each group.
The instances with $H \in \{3,6,9,12\}$ originate from 
class B of the OR-Library \cite{ORLib} and have up to 
100 nodes and 200 edges. 
For the groups with $H \in \{3,6,9\}$ 
the runtimes of all models
seem to be similar,
while for $H=12$ the new models are better.
The instances with $H \in \{5,15,25\}$ originate from 
class C of the OR-Library and have up to 
500 nodes and \numprint{12500} edges.
For the groups with $H \in \{5,25\}$ \SL seems to be faster than the 
new models. 
However, according to Table \ref{tab:resultsDIMACS:hard} 
it left one instance with $H=5$ unsolved. 
For $H=15$ the new models are faster than \SL, where the latter 
also left three instances unsolved.
%
%
Moreover, considering the worst runtime over all seven groups, 
\pop1\eurc
with 37.92 seconds for $H=25$
is the fastest, 
\pop1 (48.17 seconds for group $H=25$) the second fastest and
\pop2\eurc (54.64 seconds for group $H=25$) the third fastest model.
However, according to Table \ref{tab:resultsDIMACS:hard} 
\pop1 and \pop1\eurc do not solve all instances.
%
%
%
%
%
We can see the advantage of EURC for large $H$, too. 
%
%
%
%
%
For the groups $H\in \{15,25\}$ it speeds up \pop1,
respectively, 1.41(=31.62s/22.48s), 1.27(=48.17s/37.92s)
times, 
and \pop2,
respectively, 1.43(=44.03s/30.72s), 1.96(=107.20s/54.64s)
times on average.
%

%
%
%
%
\subsection{Strength of the LP relaxation.}

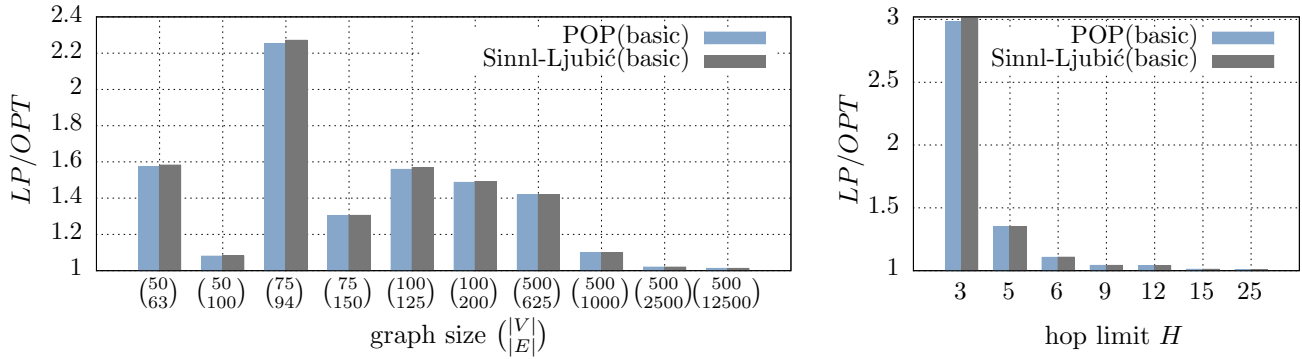
\begin{figure*}[bt]
  \centering
  \begin{tikzpicture}
    \path (-6.5,0.0) node {\includegraphics[scale=0.95]{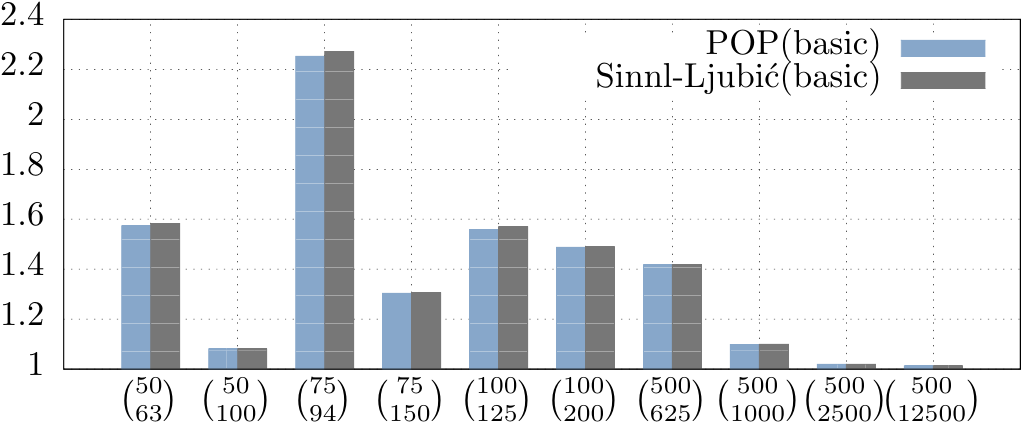}};
    \path (-11.8,0.1) node [rotate=90] {$LP/OPT$};
    \path ( 2.3,0.05) node {\includegraphics[scale=0.95]{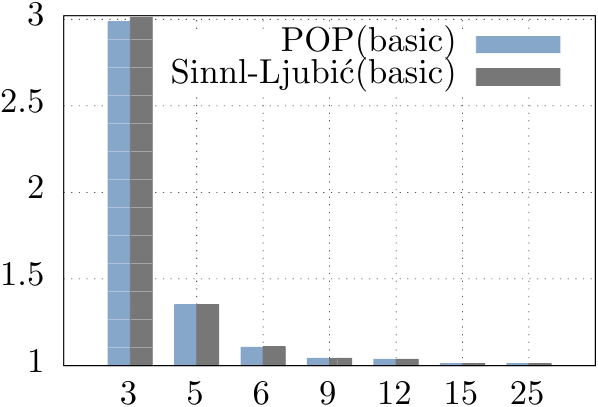}};
    \path (-0.8,0.1) node [rotate=90] {$LP/OPT$};
    \path (+2.7,-2.4) node {hop limit $H$};
    \path (-6.0,-2.4) node {graph size ${|V|}\choose{|E|}$};
  \end{tikzpicture}
  \caption{Strength ($LP/OPT$) of the basic linear programs of POP and \SL}
  \label{fig:strength:lp}
\end{figure*}

We were interested in the quality of the bounds of the LP relaxations
of both approaches.
%
%
The strength of the LP relaxation is defined as $LP/OPT$,
where $LP$ and $OPT$ are the objective values of the LP relaxation and the ILP,
respectively.
The $OPT$ values of all 414 instances are obtained using $\pop2$. 
For this test we used 
the basic model of \SL, which is 
denoted as \emph{``sNODEHOP''} in \cite{DBLP:journals/mpc/SinnlL16} 
and compared it with the basic model (\pop), (\ref{constr:root})--(\ref{constr:budget}).
The evaluations showed that 
for all of the 414 instances
our approach was at least as good as the \SL approach.
For 86 instances the basic \pop\ model was better.
%
%
%
Figure \ref{fig:strength:lp} shows the strength of both approaches 
in dependency of the graph sizes and the hop limits.
With respect to both dependencies 
the results of both approaches are almost the same, 
but the new approach is slightly better than \SL (the detailed numbers can be found in the Appendix). 
It is interesting to see that with increasing graph density 
the strength of the LP relaxation of our approach improves.
The left plot in Figure \ref{fig:strength:lp} shows that for the graphs
with 500 nodes and at least \numprint{2500} edges, the LP bound is almost 
optimum. The same is true for an increasing hop limit. Already for 
hop limit 9 the strength is almost 1.
However, Figures \ref{fig:histogram:tVE} and \ref{fig:histogram:tH}
show that with increasing graph density and hop limits
our approach gets slower. This is because 
the size of our ILP model grows. This problem can be overcome
using a \emph{cutting-plane} approach.

%
\section{Conclusion}
\label{sec:conclusion}
We presented a new binary linear program for \stprbh based on partial ordering 
of the nodes.
It has polynomial size and can 
be fed directly into a standard ILP solver. 
%
%
Using the DIMACS \cite{DIMACS11BENCHMARK} instances
we compared four variants of our ILP
with (up to our knowledge) the best known 
state-of-art algorithm 
suggested by Sinnl and Ljubi\'c 
\cite{DBLP:journals/mpc/SinnlL16}, which 
is a sophisticated \bc algorithm.
While \SL\ left 4 DIMACS instances unsolved
within a time limit of 3 hours,
our model variant \pop2 
solves
all 414 instances within a time limit of 2 hours. 
Our experiments showed that for all the instances
the strength of the LP relaxation of our basic model
dominates the \SL\ basic model.
%
%
We also introduced 
a new reduction technique for \stprbh.
The reduction decreases the running times on average 
up to 1.55 times for the largest benchmark graphs
and 1.96 times for instances with the largest hop limit.
The model \pop2\eurc, which uses the reduction, 
solves 410 of 414 instances within a time limit of 1027 seconds.
%

\bibliographystyle{plainurl}
\bibliography{stprbh_pop.bib}

\clearpage
\appendix
\section{Results of DIMACS instances}
\subsection{Objective bounds of the unsolved\\\hbox{instances by \pop1 or \pop1\eurc}}

{
  \tiny
  \setlength{\tabcolsep}{1.5pt} 
  \begin{tabular}{l|rr|r|rrr|rrr}
  \noalign{\vspace{3mm}}
    \hline
      & && & &&& && \\[-3pt]
               &       &      &    &   & \pop1 &        &   &\pop1\eurc &\\ 
     Instance  & $|V|$ &$|E|$ &Opt &lb &ub     &Time[s] &lb &ub         &Time[s]  \\ 
     \hline
     C18-10-1000-5   &500  &12500  &318.0   &318.0   &318.0   &4365.51   &318.0   &320.0   &10805.81  \\
C20-100-1000-5  &500  &12500  &4768.0  &4768.0  &4830.0  &10808.90  &4768.0  &4897.0  &10804.77  \\
C20-10-1000-5   &500  &12500  &460.0   &457.0   &469.0   &10808.75  &457.0   &470.0   &10804.85  \\

     \hline
  \end{tabular}
}    
    
\subsection{Strength of the LP relaxation}
\begin{description}
\setlength\itemsep{0pt} 
\item[$s_{POP}$:] strength (LP/OPT) of the basic model of \pop
\item[$s_{SL}$:] strength (LP/OPT) of the basic model \emph{``sNODEHOP''} of \SL
\end{description}
\scalebox{0.85}{
  \small
  \setlength{\tabcolsep}{3.5pt} 
  \begin{tabular}{l|r|r}
    \hline
      & & \\[-2pt]
     $(|V|,|E|)$     & $s_{\pop}$  & $s_{SL}$    \\ 
      & & \\[-2pt]
     \hline
      & & \\[-2pt]
	(50,63)      & 1.57234625  & 1.58200542  \\ 
 	(50,100)     & 1.08009000  & 1.08067417  \\ 
 	(75,94)      & 2.25153167  & 2.27047708  \\ 
 	(75,150)     & 1.30335042  & 1.30501250  \\ 
 	(100,125)    & 1.55867542  & 1.56920250  \\ 
 	(100,200)    & 1.48682500  & 1.48936542  \\ 
 	(500,625)    & 1.41726117  & 1.41844117  \\ 
 	(500,1000)   & 1.09830350  & 1.09940650  \\ 
 	(500,2500)   & 1.01934767  & 1.01947050  \\ 
 	(500,12500)  & 1.01247111  & 1.01249911  \\ 
     \hline
  \end{tabular}
  \hspace{12pt} 
  \begin{tabular}{l|r|r}
    \hline
      & & \\[-2pt]
     $H$     & $s_{\pop}$  & $s_{SL}$    \\ 
      & & \\[-2pt]
     \hline
      & & \\[-2pt]
      3   & 2.98347278  & 3.01212306  \\ 
      5   & 1.34991289  & 1.35152433  \\ 
      6   & 1.10631972  & 1.10689056  \\ 
      9   & 1.04166361  & 1.04172139  \\ 
      12  & 1.03708972  & 1.03708972  \\ 
      15  & 1.01063967  & 1.01065889  \\ 
      25  & 1.00852678  & 1.00852800  \\ 
     \hline
  \end{tabular}
}    

\begin{figure*} 
  \begin{minipage}[ht]{\textwidth}
\subsection{Running times for all 414 DIMACS instances}
    \begin{description}
    \setlength\itemsep{0pt} 
    \item[$|V'|$,$|E'|$:] number of nodes and edges after EURC reduction
    \item[SL:] \SL algorithm \cite{DBLP:journals/mpc/SinnlL16}
    \item[Time:] total time (in seconds) spent for the reduction, writing 
    an ILP and by the ILP-solver
    \end{description}
  \end{minipage}

  \centering
    \setlength{\tabcolsep}{2pt}
    \scriptsize
    \begin{tabular}{l|rr|rr|r|rrrr|r}
      \hline
      & && && & &&&& \\[-3pt]
               &       &      &        &       &    &\pop1 &\pop1\eurc &\pop2 &\pop2\eurc &SL\\ 
     Instance  & $|V|$ &$|E|$ & $|V'|$ &$|E'|$ &Opt &Time[s] &Time[s] &Time[s] &Time[s] &Time[s]  \\ 
      \hline
      \input{figure/table_dimacs_part1.tex}
      \hline
    \end{tabular}
\end{figure*}

\def\dimacs#1{
\begin{figure*}[ht]
  \centering
    \setlength{\tabcolsep}{2pt}

    \scriptsize
    \begin{tabular}{l|rr|rr|r|rrrr|r}
      \hline
      & && && & &&&& \\[-3pt]
               &       &      &        &       &    &\pop1 &\pop1\eurc &\pop2 &\pop2\eurc &SL\\ 
     Instance  & $|V|$ &$|E|$ & $|V'|$ &$|E'|$ &Opt &Time[s] &Time[s] &Time[s] &Time[s] &Time[s]  \\ 
      \hline
      \input{#1}
      \hline
    \end{tabular}
\end{figure*}
}

{\tiny

\begin{figure*}[ht]
  \centering
    \setlength{\tabcolsep}{2pt}

    \scriptsize
    \begin{tabular}{l|rr|rr|r|rrrr|r}
      \hline
      & && && & &&&& \\[-3pt]
               &       &      &        &       &    &\pop1 &\pop1\eurc &\pop2 &\pop2\eurc &SL\\ 
     Instance  & $|V|$ &$|E|$ & $|V'|$ &$|E'|$ &Opt &Time[s] &Time[s] &Time[s] &Time[s] &Time[s]  \\ 
      \hline
      \input{figure/table_dimacs_part2.tex}
      \hline
    \end{tabular}
\end{figure*}

\begin{figure*}[ht]
  \centering
    \setlength{\tabcolsep}{2pt}

    \scriptsize
    \begin{tabular}{l|rr|rr|r|rrrr|r}
      \hline
      & && && & &&&& \\[-3pt]
               &       &      &        &       &    &\pop1 &\pop1\eurc &\pop2 &\pop2\eurc &SL\\ 
     Instance  & $|V|$ &$|E|$ & $|V'|$ &$|E'|$ &Opt &Time[s] &Time[s] &Time[s] &Time[s] &Time[s]  \\ 
      \hline
      \input{figure/table_dimacs_part3.tex}
      \hline
    \end{tabular}
\end{figure*}

\begin{figure*}[ht]
  \centering
    \setlength{\tabcolsep}{2pt}

    \scriptsize
    \begin{tabular}{l|rr|rr|r|rrrr|r}
      \hline
      & && && & &&&& \\[-3pt]
               &       &      &        &       &    &\pop1 &\pop1\eurc &\pop2 &\pop2\eurc &SL\\ 
     Instance  & $|V|$ &$|E|$ & $|V'|$ &$|E'|$ &Opt &Time[s] &Time[s] &Time[s] &Time[s] &Time[s]  \\ 
      \hline
      \input{figure/table_dimacs_part4.tex}
      \hline
    \end{tabular}
\end{figure*}

\begin{figure*}[ht]
  \centering
    \setlength{\tabcolsep}{2pt}

    \scriptsize
    \begin{tabular}{l|rr|rr|r|rrrr|r}
      \hline
      & && && & &&&& \\[-3pt]
               &       &      &        &       &    &\pop1 &\pop1\eurc &\pop2 &\pop2\eurc &SL\\ 
     Instance  & $|V|$ &$|E|$ & $|V'|$ &$|E'|$ &Opt &Time[s] &Time[s] &Time[s] &Time[s] &Time[s]  \\ 
      \hline
      \input{figure/table_dimacs_part5.tex}
      \hline
    \end{tabular}
\end{figure*}

\begin{figure*}[ht]
  \centering
    \setlength{\tabcolsep}{2pt}

    \scriptsize
    \begin{tabular}{l|rr|rr|r|rrrr|r}
      \hline
      & && && & &&&& \\[-3pt]
               &       &      &        &       &    &\pop1 &\pop1\eurc &\pop2 &\pop2\eurc &SL\\ 
     Instance  & $|V|$ &$|E|$ & $|V'|$ &$|E'|$ &Opt &Time[s] &Time[s] &Time[s] &Time[s] &Time[s]  \\ 
      \hline
      \input{figure/table_dimacs_part6.tex}
      \hline
    \end{tabular}
\end{figure*}
 
}

\end{document}